\documentclass[pra,twocolumn, aps]{revtex4-1}
\usepackage{amsmath,comment,amsfonts,amssymb,amsthm,graphicx,bbm,enumerate,times,color}
\usepackage{mathtools}

\newtheorem{theorem}{Theorem}

\newtheorem{lemma}[theorem]{Lemma}

\newtheorem{observation}[theorem]{Observation}

\usepackage{hyperref}
\newcommand{\ro}[1]{\textcolor{black}{#1}}
\newcommand{\je}[1]{\textcolor{black}{#1}}

\newcommand{\he}[1]{{\color{black}#1}}

\newcommand{\GP}{{\text{GP}}}

\newcommand{\TC}{{\text{TC}}}
\newcommand{\TO}{{\text{TO}}}
\newcommand{\OPT}{{\text{opt}}}

\newcommand{\e}{{\rm e}}

\newcommand{\RR}{\mathbb{R}}

\newcommand{\id}{\mathbbm{1}}
\newcommand{\mc}[1]{\mathcal{#1}}

\newcommand{\proj}[1]{\vert #1\rangle\!\langle#1 \vert}

\newcommand{\norm}[1]{\left\Vert #1 \right\Vert}

\newcommand{\Tr}{\operatorname{tr}}
\newcommand{\tr}{\Tr}
\newcommand{\hams}{\mc{H}} 
\newcommand{\localhams}{\hams_{\mathrm{loc}}} 
\newcommand{\maps}{\mc{M}} 
\newcommand{\mm}{\Omega} 
\newcommand{\relent}[2]{D\left(#1\!\mid\mid\!#2\right)} 

\newcommand{\fu}{Dahlem Center for Complex Quantum Systems, Freie Universit{\"a}t Berlin, 14195 Berlin, Germany}

\begin{document}
\title{{Second laws} under control restrictions}

\author{H.\ Wilming,  R.\ Gallego, and J.\ Eisert}

\affiliation{\fu}

\begin{abstract}
The second law of thermodynamics, formulated as an ultimate bound on the maximum extractable work, has been rigorously derived in multiple scenarios. However, the unavoidable limitations that emerge due to the lack of control on small systems are often disregarded when deriving such bounds, which is specifically important in the context of quantum thermodynamics. Here, we study the maximum extractable work with limited control over the working system and its interaction with the heat bath. We derive a general second law when the set of accessible Hamiltonians of the working-system is arbitrarily restricted. We then apply our bound to particular scenarios that are important in realistic implementations: limitations on the maximum energy gap and local control over many-body systems. We hence demonstrate in what precise way the lack of control affects the second law. In particular, contrary to the unrestricted case, we show that the optimal work extraction is not achieved by simple thermal contacts. Our results do not only generalize the second law to scenarios of practical relevance, but also take first steps in the direction of local thermodynamics.
\end{abstract}

\maketitle
\section{Introduction}
Recently, there has been much progress in our understanding of thermodynamics {in} the quantum domain. Particular emphasis has been put on deriving fundamental bounds on how much work can {precisely} be extracted in a thermodynamical process under {meaningful} assumptions. These bounds are indeed understood as the second law of thermodynamics in Thomson's formulation. Progress has been made {specifically} in {deriving} bounds that apply to the most general scenarios, including non-equilibrium states \cite{Alicki,Alicki2,Esposito,Linden10,Brunner12,Anders13,OOE}, account for work fluctuations \cite{Jarzynski97,Egloff12,Aberg13}, and more general classes of interactions between the working system and the heat bath \cite{Resource,ThermoMaj,ThermoMaj2,SecondLaws,SL2,Renes14,Gallego2013}.

Clearly, in {any} real experiment, one will face various specific limitations on the allowed operations, 
{rendering it} in general impossible to saturate the second law. {In} most cases these are limitations that pertain to the particular technological implementation of the thermal machine and do not encode any fundamental limitation in the same way the second law does. Nonetheless, specially when dealing with systems at the scale where quantum effects become relevant, {it seems imperative to consider} general classes of limitations on the control of the systems used in the protocol of work extraction. These limitations are not 
{specific to} the substrate or {the} technology employed, but {are of} 
fundamental nature as they will {naturally} emerge in any conceivable implementation. Thus, the tighter bounds derived on the extractable work when accounting for such limitations should indeed be understood as quantum versions of the second law, in the sense that they encode the ultimate bounds conceivable when {making use of} small systems. 

\begin{figure}[t!b]
\label{fig:bit}
\centering
\includegraphics[width=4.6cm]{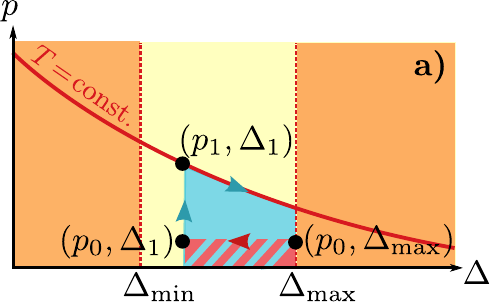}
\includegraphics[width=4.6cm]{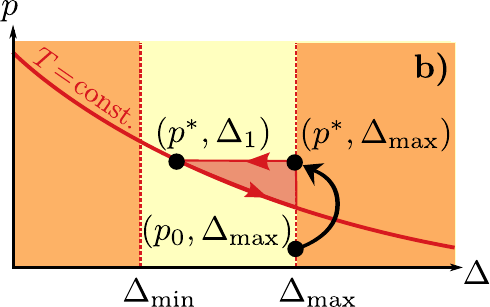}
\caption{An example exhibiting an energy-constraint.  $p$ is the excitation probability of a classical bit and $\Delta$ is the energy in the excited state. The dark orange region is forbidden. 
Upper figure: Example of a general protocol using WTC. The total work extracted work (solid blue region) is negative.
Lower figure: The black arrow denotes an initial thermalizing map in $\mc{M}_{\text{TO}}$. Then, the transition $(p^*,\Delta_{\text{max}})\rightarrow (p^*,\Delta_1)$ takes place, followed by an isothermal transition back to $\Delta_{\text{max}}$. The total work extracted (red solid area) is positive.}
\end{figure}

In this work we introduce a general framework {allowing} to study bounds on work extraction when limitations are imposed on the control of the systems involved. These limitations are of two classes, namely 
\begin{itemize}
\item[i)] Restrictions on the accessible set of Hamiltonians of the working system and 
\item[ii)] restrictions on the allowed interactions of the working sytems with the heat bath available. 
\end{itemize}
Regarding the class i), we will first introduce a general framework which considers an arbitrary family $\hams$ of allowed Hamiltonians. Then we will 
{turn to studying} two fundamental limitations that emerge ubiquitously: 
\begin{itemize}
\item[i.a)] Bounds on the maximum energy gap of the Hamiltonians in $\hams$ and 
\item[i.b)] having only access to change the local Hamiltonians of many-body systems. 
\end{itemize}
These limitations are introduced in a completely generic fashion, without making reference to any particular model. Thus, they encode fundamental limitations known to be present in any implementation.

Regarding the limitations of the class ii), we will consider different scenarios that model the degree of control that the experimenter has over the {bath's} degrees of freedom. The simplest model, corresponding to {the absence of any} fine-tuned control, is thermal contact (TC) which merely thermalises the system at the bath's temperature. At the opposite end {of the hierarchy of meaningful models of system-bath contacts} are the so-called thermal operations (TO), where full control is assumed on all the degrees of freedom of the heat {bath} \cite{ThermoMaj,ThermoMaj2}.

Our main result is to derive a general second law for the case of most physical relevance, where constraints i) and ii) come {into play}. We {identify} a bound to the maximum extractable work when the Hamiltonians are restricted to a class $\hams$ and the interaction with the bath is modeled by TC. Our result applies to initial and final states out of equilibrium and non-cyclic processes. {This bound is then used} to study how the different limitations interplay with each other when they are relaxed. In particular, it is has {previously} been  shown that the limitations of the class ii) alone do not tighten the bounds on work extraction \cite{Aberg13}. This observation -- which encodes the surprising fact that having {complete} control over the heat bath offers no advantage -- has diminished the interest on comparing different models for the degree of control over the bath, since it is seemed irrelevant to the problem of work extraction. However, we show {that the situation is radically different in the presence of limitations}: control over the bath's degree of freedom is advantageous for work extraction if other restrictions are simultaneously imposed, namely, i.a) and i.b) {\cite{RossnagelComment}}.

\section{Set-up}
We consider a finite-dimensional quantum system, 
initially described by a pair $\big(\rho_0,H_0\big)$ of an initial quantum state $\rho_0$ and an initial Hamiltonian $H_0$. 
This system can undergo an arbitrary \emph{protocol} consisting of two types of operations.
The first kind of operation is a \emph{unitary time-evolution} for some time $T:=t_2-t_1$, corresponding to a transformation
\begin{equation}\label{eq:unitevol}
\big(\rho_{t_1},H_{t_1}\big)\mapsto \big(U_T\rho_{t_1} U_T^\dagger,H_{t_2}\big)=:\big(\rho_{t_2},H_{t_2}\big),
\end{equation} 
$U_T$ {reflecting} the unitary time-evolution under a time-dependent Hamiltonian $t\mapsto H_t$ from $t_1$ to $t_2$. 
{We will now} introduce {a} first class of restrictions on the {control}: We model the lack of control over the system by restricting the time-dependent Hamiltonian $H_t$ to a set $\hams(H_0)$ which in general depends on the initial Hamiltonian, i.e., $H_t\in\hams(H_0)$ for all times \cite{PathConnected}.
We will assume that every such operation costs an amount of average work {given} by the difference of the system's average energy 
$\langle W\rangle = \Tr(\rho_{t_1}H_{t_1}) - \Tr(\rho_{t_2}H_{t_2})$ \cite{footnote1}.

The second kind of allowed operations models the coupling of the system to heat baths at some \emph{fixed} inverse temperature $\beta>0$. 
Given that the work cost at subsequent operations of the kind \eqref{eq:unitevol} does not depend on the state of the bath, but on the state of the system, it suffices to describe the thermal contact as an effective quantum map on the system:
We allow that, at any time $t$, a quantum channel $\mc{G}_t$ can be applied, resulting in a transformation of the form
\begin{align}\label{eq:mapG}
\big(\rho_t,H_t\big) \mapsto \left(\mc{G}_{t}(\rho_t),H_t\right).
\end{align}
In general we will refer to such quantum channels as \emph{thermalizing maps}. $\mc{G}_t$ models the effective time-evolution of the system when it is put in contact with a bath at time $t$ and does not cost any work. The specific form of $\mc{G}_t$ will depend not only on $H_t$,  but also on the heat bath that is considered and the way one makes it interact with the system. This is mathematically expressed by restricting the map to be in a particular set, i.e. $\mc{G}_t \in \maps\big(H_t\big)$. The choice of $\maps$ encodes in which way we model the bath and the degree of control over it to implement a given interaction with the system. 
The minimal assumption on the set of maps $\maps$, is that any $\mc{G}_t \in \maps\big(H_t\big)$ leaves the thermal state of the system invariant, 
\begin{equation}
\label{eq:gibbspreserving}
\mc{G}_t\left(\omega_{H_t}\right)=\omega_{H_t},
\end{equation}
where 
\begin{equation}
	\omega_{H_t}=\exp(-\beta H_t)/Z 
\end{equation}
denotes the Gibbs state of the system at inverse temperature $\beta>0$. We will also express this by saying that the thermalizing maps are \emph{Gibbs-preserving} \cite{ThermoMaj}. Condition \eqref{eq:gibbspreserving} is necessary since if it would be violated, one could create states out of equilibrium from thermal states, which would make work extraction trivial in the sense that it can be performed without expenditure of resources. Nonetheless, even if condition \eqref{eq:gibbspreserving} is satisfied, implementing an arbitrary Gibbs-preserving map requires in general high degree of control over the degrees of freedom of the thermal bath, which may be in most situations of interest unrealistic. The second kind of restrictions that we consider will be concerned with physically meaningful limitations on such control, which will determine in turn a restricted set of maps $\maps$.

We now define a protocol of work extraction $\mc{P}$ as an arbitrary finite
sequence of operations as in Eqs.\ \eqref{eq:unitevol} and \eqref{eq:mapG}, from an initial {condition} $p_0:=(\rho_0,H_0)$ to a final $p_f:=(\rho_{t_f},H_{t_f})$.
The total amount of work extracted will be the sum of the work extracted in each operation and it will depend on both the initial and final condition and the protocol $\mc{P}$.
This protocol will involve allowed operations  that respect the potential physical constraints on $\hams$ and makes use of thermalizing maps ${\cal M}$; 
we denote the set of all protocols fulfilling such constraints by $\mathcal{P}_{\hams,\mathcal{M}}$.
These limitations have an impact on the optimal 
{work that can be extracted}
\begin{equation}\label{eq:optwork}
	\langle W \rangle_{\OPT}^{\hams,\mathcal{M}}(p_0,p_f):=\sup_{\mc{P}\in \mathcal{P}_{\hams,\mathcal{M}} } \langle W \rangle \big(\mc{P},p_0,p_f\big).
\end{equation}
Usually, it is of particular relevance to bound the maximum work that can be extracted from a given initial condition $p_0$ in a cyclic protocol where $H_{t_f}=H_0$ because of its relation with the formulation of the second law; in such case we will simply use the notation
\begin{equation}
\langle W \rangle_{\OPT}^{\hams,\mathcal{M}}(p_0):=\sup_{p_f=(\rho_f,H_0)}\langle W \rangle_{\OPT}^{\hams,\mathcal{M}}(p_0,p_f).
\end{equation}

\section{Thermalising maps and second law}
We now turn to discussing the considered models for the set of \emph{thermalizing maps} $\maps$, which are physically meaningful and standard in the quantum thermodynamics literature.
The first such operation models simple thermal contact (TC) with the bath. Formally, it is captured as
\begin{equation}\label{def:wtcmap}
\maps_{\TC}\big(H_t\big)=\lbrace \mc{G}_t : \mc{G}_t(\rho)=
\omega_{H_t} \rbrace,
\end{equation}
with the Gibbs state $\omega_{H_t}=\exp(-\beta H_t)/Z_{H_t}$, where $Z_{H_t}=\tr \big(\exp (-\beta H_t)\big)$.
This set of maps merely contains one single element that maps every state to the Gibbs state at the bath's temperature. This is the map implemented by a sufficiently long time evolution under a sufficiently weak arbitrary system-bath interaction \cite{Riera12}. As such, it requires no control over the degrees of freedom of the heat bath.

The second class of operations that we consider is constituted by thermal operations (TO). They 
depend on the Hamiltonian $H_B$ of the particular 
bath at hand and are defined as
\begin{equation}\label{def:tomaps}
\maps_{\TO}\big(H_t\big)=\lbrace\mc{G}_t : \mc{G}_t(\rho)=
\tr_B( U (\omega_{H_B} \otimes \rho) U^{\dagger}) \rbrace
\end{equation}
for any unitary $U$ so that $[U,H_B+H_t]=0$ \cite{ThermoMaj2,Resource}. Thermal operations have originally been introduced in the framework of \emph{single-shot} thermodynamics, albeit for the same reason that we use them here: they model precisely the case of arbitrary control over system and heat bath when energy and entropy are conserved exactly. 
For convenience we will collect all thermalizing maps fulfilling the Gibbs-preserving condition w.r.t.\ $H_t$ in the set $\maps_{\GP}\big(H_t\big)$. Then we have {the strict inclusions \cite{Faist2014}} $\maps_{TC}\big(H_t\big)\subset\maps_{TO}\big(H_t\big)\subset \maps_{GP}\big(H_t\big)$.
\he{We are now in position to state the second law in terms of work-extraction as
\begin{equation}
\langle W\rangle^{\GP}_{\OPT}(p_0,p_f) \leq F(\rho_0,H_0)-F(\rho_f,H_f),
\end{equation}
with $p_0=(\rho_0,H_0)$ and $p_f=(\rho_f,H_f)$. This bound has been shown before \cite{Aberg14}, but for convenience of the reader we give a proof in Section \ref{sec:aberg}.}

In the remainder of this work we will study how restrictions on the sets $\hams(H_0)$ and $\maps(H_t)$ interplay and influence the maximum value for work extraction $\langle W \rangle_{\OPT}^{\hams,\mathcal{M}}$ in Eq.\ \eqref{eq:optwork}. In particular, we will now derive the form of the second law when both restrictions on $\maps$ and $\hams$ come into play.

\section{General second law under control restrictions} 
We will now consider the scenario where control is restricted i) to an arbitrary family of Hamiltonians $\hams$ on the system and ii) no control over the degrees of freedom of the heat bath. We derive the most general form of a the second law in the sense that it allows for initial and final states out of equilibrium and non-cyclic processes. Let us define $\mc{U}[H_0]$ to be the unitary group generated by arbitrary time-evolutions under time-dependent Hamiltonians $H_t\in \hams(H_0)$. The non-equilibrium free energy of $(\rho,H)$ is given by $F(\rho,H) = \Tr(\rho H) - S(\rho)/\beta$. 
We define the (von Neumann) free energy difference to the Gibbs state by $\Delta F(\rho,H) = F(\rho,H)-F(\omega_H,H)$.
The function $\Delta F$ quantifies how far out of equilibrium $(\rho,H)$ is. 

\begin{theorem}[Second law under control restrictions]\label{thm:maxworkbound} The maximum work that can be extracted in a protocol $\mc{P}$ from {the} pair $p_0=(\rho_0,H_0)$ to $p_f=(\rho_f,H_f)$ 
{by} combining time-dependent Hamiltonians from $\hams(H_0)$ and thermalizing maps of the form $\mc{G}_{t}\in \maps_{\TC}\big(H_t\big)$ is bounded by
 \begin{eqnarray}
\nonumber \langle W \rangle^{\hams,\TC}_{\OPT}
 (p_0,p_f) &\leq&  F(\rho_0,H_0)\,-F(\rho_f,H_f)\\
\label{eq:maxworkbound} &-&\!\!\!\!\inf_{{H_t\in \hams(H_0),\, \sigma\in \mathcal{U}[H_0](\rho_0)}}\!\!\!\!\!\!\Delta F(\sigma,H_t),
\end{eqnarray}
where $\mathcal{U}[H_0](\rho_0)$ is the unitary orbit of $\rho_0$ with respect to $\mathcal{U}[H_0]$. 
\je{The bound can be saturated arbitrarily well.}
\end{theorem}
\he{\begin{proof}
The initial step of the protocol must be given by unitary dynamics, since otherwise no work can be extracted: If we first thermalise, we effectively start from the Gibbs state and no work can be extracted from a Gibbs state at the same temperature as the heat bath. We will therefore assume that the protocol starts with unitary dynamics. 
We will write $(\rho_{j},H_{j})$ for the quantum state and Hamiltonian prior to the $j$-th thermalisation. We hence have that $\rho_{1}=U_{1}\rho_0 U_{1}^{-1}$ and 
\begin{equation}
	\rho_{j+1}=U_{j+1}\omega_{H_{j}}U_{j+1}^{-1}=:\omega_{H_{j}}^U 
\end{equation}	
	where $U_{j+1}$ is the unitary that implements the time-dependent Hamiltonian dynamics starting after the $j$-th thermalisation and ending right before the $j+1$-th thermalisation. We assume that our protocol consists of $N$ thermalisations and that, without loss of generality, the last step in the protocol is unitary dynamics. 
The total work is then given by
\begin{eqnarray}\label{eq:workproof}
	\langle W  \rangle^{\TC} &= &\tr\big(\rho_{{0}} H_{{0}}\big)-\tr\big(\rho_{1} H_{1}\big)
	\nonumber\\
	&+&\sum_{j=1}^{N-1}\left(\tr\big(\omega_{H_{j}} H_{j})-\tr\big(\omega_{H_{j}}^UH_{j+1})\right)\nonumber\\
        &+&\tr\big(\omega_{H_N}H_N\big) - \tr\big(\rho_fH_f\big),
\end{eqnarray}
with  {$\rho_f$ being defined as} $\rho_f=U_{N+1}\omega_{H_N}U_{N+1}^{-1}$. We will now reformulate this expression in terms of the free energy. {Given any two Hamiltonians $H$ and $H'$ we have
\begin{eqnarray}
\tr (\omega_H H) - \tr (\omega_{H}^U H')&=& F(\omega_H,H)-F(\omega_H^U,H')\nonumber\\
&\leq& F(\omega_H,H) -F(\omega_{H'},H').
\end{eqnarray}
The first equality follows from applying the definition 
\begin{equation}
	F(\rho,H)=\tr(\rho H)-S(\rho)/\beta
\end{equation}	
 and using that $\omega_H$ and $\omega_H^U$ have the same entropy. The inequality then follows from the fact that $F(\omega_H,H)\leq F(\rho,H)$ for every $(\rho,H)$.
Now let us apply this inequality to \eqref{eq:workproof}, to obtain
\begin{eqnarray}
\langle W  \rangle^{\TC}&\leq &F\big(\rho_{{0}},H_{{0}}\big)-F\big(\rho_{1},H_{1}\big)\nonumber\\ &+&\sum_{j=1}^{N}
\left( F(\omega_{H_{j}},{H}_{j})-F(\omega_{H_{j+1}},H_{j+1})\right)\nonumber\\
&+&F\big(\omega_{H_N},H_N\big) - F\big(\rho_f,H_f\big),\nonumber\\
&=&F(\rho_{{0}},H_{{0}})-F(\rho_f,H_f)\nonumber\\\label{eq:calcwork1} &-&\left(F(\rho_{1},H_{1})-F(\omega_{H_{1}},H_{1})\right),
\end{eqnarray}
where Eq.\ \eqref{eq:calcwork1}} follows  again from the definition of $F$ and $S(\rho_{{0}})=S(\rho_{1})$. Since 
\begin{equation}
	F(\omega_{H},H)\leq F(\rho,H), 
\end{equation}
we can minimise $\rho_{1}$ over all possible states in $\mathcal{U}[H_{{0}}](\rho_{{0}})$ and $H_{1}$ over all Hamiltonians in $\hams(H_{{0}})$ to get the desired bound.

To show that the bound is achievable let us recall \cite{Gallego2013,Aberg13} that a slow change of Hamiltonians from $H_{1}$ to $H_{N}$ -- slow in the sense that the system is always in the thermal state for the given Hamiltonian at that time -- has a work cost of $F(\omega_{H_{1}},H_{1})-F(\omega_{H_N},H_N)$. Such a protocol is achievable to arbitrary accuracy. 
\end{proof}}
The bound given by \eqref{eq:maxworkbound} can be nicely interpreted by first noting that if $\hams(H_0)$ is unrestricted, the infimum-term in the bound is zero as any state can be approximated to arbitrary accuracy by a Gibbs state. Hence, we obtain the usual second law given by a difference of the non-equilibrium free energies \cite{Alicki,Alicki2,Esposito,Anders13,Aberg13,Gallego2013}. Thus, the infimum-term in \eqref{eq:maxworkbound} should be interpreted as a penalty 
that emerges as a consequence of the limitations on control. Lastly, {in order to let the penalty term vanish, it} suffices to take $\hams$ unrestricted (while $\maps$ is restricted); this fact may be lead one to the conclusion that {the} limitations on the bath control given by $\maps$ are indeed irrelevant for work extraction. We will now show that this is general not the case.

\section{Restrictions on bath control only}
{The previous findings have a remarkable consequence that has,
 to the best of our our knowledge, not been noted before:} If no other restrictions are put {onto} the valid operations, having control over the degrees of freedom of the heat bath is useless when it comes to work extraction. This is an astonishing fact, given that this holds regardless of the size, Hamiltonian, model or substrate that we use to describe the heat bath. Thus, controlling its degrees of freedom may be an enormous technological challenge that increases the valid set of operations, while not being useful at all for the main task in thermodynamics, i.e., 
 work extraction:
   
\begin{observation}[Universality of WTC in unrestricted setting] If $\hams(H_0)$ is the set of all Hamiltonians, protocols employing thermalizing maps of the form $\mc{G}_{t}\in \maps_{\TC}\big(H_t\big)$
or $\mc{G}_{t}\in \maps_\GP\big(H_t\big)$, both achieve the same maximum work. That is, $\langle W \rangle^{\TC}_{\OPT}(p_0,p_f) =
\langle W \rangle^{\GP}_{\OPT}(p_0,p_f)$.\label{obs:obs}
\end{observation}
This follows straightforwardly from the results in Ref.\ \cite{Aberg13} {which} state that the maximum work by using Gibbs-preserving maps is given by $\langle W \rangle^{\GP}_{\OPT}(\rho_0,H_0)
= \Delta F(\rho_{{0}},H_{{0}})$; together with the usual second law for non equilibrium states in the presence of thermal contact $\langle W \rangle^{\TC}_{\OPT}(\rho_0,H_0)
= \Delta F(\rho_{{0}},H_{{0}})$, as it follows also from \eqref{eq:maxworkbound}.
Observation \ref{obs:obs} should not be confused with the {\it maximum-work principle} \cite{MaxEnt}. We will now discuss particular examples and show that indeed the universality of TC for work extraction is only 
valid when no other restrictions are imposed, inequivalent with Ref.\ \cite{Allahverdyan}, where the maximum-work principle is shown to not apply \he{in certain situations}.

\section{Breakdown of TC-universality} 
We will now provide a general argument allowing to understand why thermal contact ceases to be universal within scenarios where restrictions on $\hams$ are imposed. To understand it, it is necessary to recapitulate what the optimal cyclic protocol in the TC-setting without constraints {on the Hamiltonians} does. {We denote} such an optimal protocol by $\mc{P}^*_{\TC}(\rho_0,H_0)$. It consists of two parts. The first is to apply a transformation $(\rho_0,H_0) \mapsto (\rho_0,H_1)$, with $H_1$ such that $\rho_0=\omega_{H_1}$ (to arbitrary accuracy). Since the initial state $\rho_0$ is already in the Gibbs form for Hamiltonian $H_1$, putting it in WTC with the bath after this first step does not produce any heat dissipation. 
The second part is a sequence of transformations of the kind $\eqref{eq:unitevol}$, followed each by TC with the bath, so that the Hamiltonian goes back from $H_1$ to $H_0$. In the limit of infinite number of
{steps}, the second part is a quasi-static transformation along an iso-thermal. It is easy to verify that $\mc{P}^*_{\TC}(\rho_0,H_0)$ extracts  $\Delta F (\rho_0,H_0)$ as work \cite{Aberg13}. 

In the case where $\hams$ does constrain the set of allowed Hamiltonians, the first step of the optimal protocol is in general not possible. For a given $(\rho_0,H_0)$ there might be no Hamiltonian $H_1$ in $\mc{H}(H_0)$ such that $\rho_0=\omega_{H_1}$. An example is the scenario in which the Hamiltonians in $\hams(H_0)$ are all diagonal in the same basis but $\rho_0$ has a lot of coherences in this basis. This effect 
gives rise to the penalty term in Eq.\ (\ref{eq:maxworkbound}), which can, in the most drastic cases, diminish the maximal amount of extracted work to zero.

Now consider the setting where the bath is modelled by more general thermalizing maps, say $\maps_{\TO}(H_t)$, and assume there is indeed no Hamiltonian $H_1\in\hams(H_0)$ such that $\rho_0=\omega_{H_1}$. Then the free energy bound cannot be saturated by a protocol using weak thermal contact. In this setting, however, we may be able to use a thermal map from $\maps_{\TO}(H_0)$ to map $\rho_0$ to some state $\rho_1$ for which now there \emph{does exist} a corresponding Hamiltonian $H_1$ with $\rho_1=\omega_{H_1}$. If this is the case, we can, after having applied the map and by using the optimal protocol, extract an amount of work given by $\Delta F(\omega_{H_1},H_0)$. As long as the heat lost during this process, given by $F(\rho_0,H_0)-F(\rho_1,H_0)$, is smaller than the penalty term in Eq.\ (\ref{eq:maxworkbound}) we will be able to extract more work using thermal maps from $\maps_{\TO}(H_t)$ than by just using thermal contact. 
In the following, we use the reasoning from above to show that there are natural scenarios and initial conditions where \emph{no} work can be extracted using TC but work can be extracted using thermal operations. We will now illustrate this phenomena in two particularly relevant examples of physical restrictions.

\subsection{Example I: Bounds on the operator norm}

We first consider the scenario involving a bound on the operator-norm of the Hamiltonian. To be concrete, \ro{we consider a two-level system and Hamiltonian restrictions so that 
\begin{equation}	
\hams_{{\norm{H}}}(H_0):=\{H\ | \ \Delta_{\text{min}} \leq\norm{H}_{\infty} \leq \Delta_{\text{max}}, E=0\}, 
\end{equation}
where $E$ denotes the ground-state energy. We will now show that there exist initial configurations $(\rho_0,H_0)$ from which no positive work can be extracted using TC. 
Consider 
\begin{equation}
\rho_0=p_0 \proj{1} + (1-p_0)\proj{0} 
\end{equation}
with $p_0\leq \exp (-\beta \Delta_{\max})/{Z(\Delta_{\max})}$ and $H_0=\Delta_{\text{max}} \proj{1}$. In Section \ref{sec:app:operator-norm} we show that the optimal work $\langle W \rangle^{\TC}_{\OPT}$ can always be achieved by a protocol such that {the state and Hamiltonian before the first TC with the bath, denoted by $(\rho_{1},H_{1})$}, are diagonal in the same basis and $\rho_1=\rho_0$. Hence, it suffices to show that no work can be extracted if 
\begin{equation}
\hams(H_0)=\{\Delta_{1} \proj{1} \ : \ \Delta_{\text{min}} \leq \Delta_{1} \leq \Delta_{\text{max}} \}
\end{equation}. 
In this case, the problem is fully classical and system and Hamiltonian are described, before the first TC, by $(p_{0},\Delta_{1})$. }

\ro{The optimal protocol can be divided into two parts as depicted in Fig.\ 1 a) . First, a
  unitary evolution of the kind \he{of Eq.\ (1)}, which takes $(p_0,\Delta_0) \mapsto (p_1,\Delta_1)=(p_0,\Delta_1)$. 
Let us denote the work in this step by {$\langle W \rangle_1:=p_0(\Delta_{\text{max}}-\Delta_0)$}. The second part of the protocol contains all the operations applied after the first TC with the bath. The work extracted in the second part is denoted by $\langle W \rangle_2$. It is bounded by the free energy difference $\langle W\rangle _2 \leq F(\omega_{H_1},H_1)-F(\omega_{H_0},H_0)$. One can easily show that the total work fulfils:
\begin{align}
{\langle W \rangle_{\mathrm{opt}}^{\norm{H},\mathrm{TC}}(\rho_0,H_0)} = \langle W \rangle_1+\langle W \rangle_2 \leq 0.
\end{align}}

\ro{Next we show that if the same restrictions are imposed on the allowed Hamiltonians, but the thermalizing maps are given by $\maps_{{\TO}}$ {(or the more general $\maps_{{\GP}}$) then }there are situations in which one can extract a positive amount of work, as illustrated in Fig. 1 b). The key idea is that there exist maps {$\mc{G}_{t=0} \in \maps_{{\TO}}$} that map the initial state to another state $\rho^*$ given by 
\begin{equation}\label{eq:aht}
	\rho^*=p^*\proj{1}+(1-p^*)\proj{0},
\end{equation}
 with $p^*\geq \exp (-\beta
\Delta_{\text{max}})/{Z(\Delta_{\text{max}})}$. This is shown in Section \ref{sec:AHT} for continuity of the proof. Such a state $\rho^*$ has less free energy
than $\rho_0$, but a finite difference $\Delta F(\rho^*,H_0)> 0$ of
free energy to the Gibbs state. Once the system is in the state
$\rho^*$ one can simply apply the optimal protocol
$\mc{P}^*_{\TC}(\rho^*,H_0)$. This is possible, since now there
exists a Hamiltonian $H_1\in \hams(H_0)$ such that
$\rho^*=\omega_{H_1}$. {Denoting the whole protocol by $\mathcal{P'}$
  (that is, the composition of $\mc{G}_{t=0}$ followed by
  $\mc{P}^*_{\TC}(\rho^*,H_0)$), one can extract an amount of work
  \begin{equation}
  	\langle W\rangle(\mathcal{P'},\rho_0,H_0)= \Delta
  F(\rho^*,H_0)>0. 
  \end{equation}
  Clearly, $\mathcal{P'}$ is a protocol respecting
  the constraints of $\mathcal{H}_{\norm{H}}$ and $\mathcal{M}_{\TO}$,
  thus recalling \he{Eq.\ (4)} one finds that
\begin{align}
 \langle W\rangle_{\mathrm{opt}}^{{\norm{H}},\mathrm{TO}} (\rho_0,H_0) > \langle W\rangle_{\mathrm{opt}}^{{\norm{H}},\mathrm{TC}}(\rho_0,H_0)=0. 
\end{align}}
}
We conclude that if the set of allowed Hamiltonians is restricted by a {norm} upper bound, TO can extract a positive amount of work in some situations where TC cannot. Hence, thermal contact is not universal in this restricted scenario.

\subsubsection{Anomalous heat transfer}\label{sec:AHT}
Here we show that it is possible to perform a transition as indicated in Eq. \eqref{eq:aht}. Apart from its use in the previous section, we believe this result is interesting in its own right. Let us note that the existence of the GP-map, which maps the initial state to $\rho^*$, has a profound implication: It implies that a heat-bath at temperature $T$ can heat a colder system at a temperature $T_0<T$ to a temperature $T_1$ \emph{larger than $T$} while the total energy is conserved and the second law of thermodynamics, expressed in terms of free energies, is not violated. Moreover, it turns out that $T_1$ can be made larger by \emph{reducing} $T_0$: The colder the system is initially, the hotter it is after the GP-map. Motivated by this observation, we call this effect \emph{anomalous heat transfer}.

\he{
So let us show that there exist indeed maps in {$\mathcal{M}_{\TO}$}, which display anomalous heat transfer. That is, map a state of a classical bit with excitation probability less than the thermal one to one with a higher excitation probability than the thermal one. First note, that for pairs $(\rho,H)$ which are diagonal in the same basis, the image of the set of GP-maps and TO coincide \cite{Renes14,Faist2014}. 
Hence, it suffices to show the existence of a GP-map that performs the task above. For a classical bit, with Hamiltonian $\Delta \proj{1}$, the Gibbs-distribution is given by
\begin{align}
\omega_\Delta &= \frac{1}{1+\exp(-\beta \Delta)}\left(\exp(-\beta \Delta),1\right)\nonumber\\
 &= \left(\frac{1}{1+\exp(\beta\Delta)},\frac{1}{1+\exp(-\beta\Delta)}\right).
\end{align}
A GP-map is simply given by a stochastic matrix which has $\omega_\Delta$ as fixed point. All such matrices can be written as 
\begin{align}
\mc{G}_\Delta^r = \begin{pmatrix} r& (1-r)\exp(-\beta \Delta)\\ 1-r & 1-(1-r)\exp(-\beta \Delta)\end{pmatrix},
\end{align}
with $r\in [0,1]$. If $p=(p_e, 1-p_e)$ is an input distribution, than the output of the GP-map has excitation probability
\begin{align}
(\mc{G}_\Delta^rp)_1= \exp(-\beta \Delta)(1-r)(1-p_e)+r p_e. 
\end{align}
In particular for $r=0$ we obtain an excitation probability $\exp(-\beta \Delta)(1-p_e)$ which can easily be larger than $1~/~(1~+~\exp(\beta\Delta))$.}

\subsection{Example II: Local restrictions on a many-body system}
The second scenario that we consider encodes the incapability to change interactions in multi-partite systems. 
We consider the  scenario where multi-partite (interacting) Hamiltonians can only be manipulated locally, 
\begin{equation}
\localhams(H_0):=\bigl\{ H_0 + X \mid X = \sum_i H_i\bigr\},
\end{equation} 
with $H_i$ supported on subsystem labeled $i$. 
We will now prove the existence of a gap between the optimal work extraction within protocols using thermal contact and protocols thermal operations. As a step towards this we first show the following result, which is interesting in its own right.

\begin{theorem}[Passive states for weak thermal contact]\label{thm:passive}
There exist multipartite systems with initial conditions $p_0=(\rho_0,H_0)$ such that $\Delta F(\rho_0,H_0)>0$, from which \emph{no work} can be extracted in a cyclic process by thermal contact and $\hams(H_0)=\localhams(H_0)$. More {precisely,} $\langle W \rangle_{\OPT}^{\localhams,\TC}(p_0)\leq 0$.
\end{theorem}
\he{The proof of the theorem relies on the Peierls-Bogoliubov inequality \cite{PeierlsBogoliubov}, which we formulate in a way suitable for us in the subsequent Lemma and also prove for the convenience of the reader.

\begin{lemma}[Upper bound to free energies]\label{Peierls}
For any two Hamiltonians $A,B$,
$F(\omega_{A+B},A+B) \leq F(\omega_A,A)+\tr(\omega_AB)$.

\begin{proof}
Let 
\begin{equation}
Z_A=\tr(\exp(-\beta A))
\end{equation}
 and consider $\relent{\omega_A}{\omega_{A+B}}$, which is always positive. For any two Hamiltonians $A,B$ we have
\begin{align}
\tr(\omega_A\log \omega_B) = -\beta \tr(\omega_A B) - \log Z_B. 
\end{align}
Therefore, we get for the relative entropy
\begin{align}
\relent{\omega_A}{\omega_{A+B}} &= -\beta \Tr(\omega_AA) - \log Z_A \nonumber\\
& + \beta \tr(\omega_A(A+B))+\log Z_{A+B}\geq 0.
\end{align}
Cancelling the terms $\beta \tr(\omega_AA)$ and recognising that $F(\omega_A,A)=-\beta^{-1}\log Z_A$ finishes the proof.
\end{proof}
\end{lemma}

\begin{proof} {\it (Of Theorem~\ref{thm:passive})}
Without loss of generality we will assume that all Hamiltonians are traceless. Now imagine that there exists an $H_0$ such that
$F(\omega_{H_0},H_0) \geq F(\omega_{H_t},H_t)$
for all $H_t\in \localhams(H_0)$.
Since the maximally mixed state $\mm$ is invariant under arbitrary unitaries, Thm.\ \ref{thm:maxworkbound} implies that no work can be extracted from $(\mm, H_0)$. 
We now prove that for $\localhams(H_0)$ such special Hamiltonians indeed exist. For any $H_t\in \localhams(H_0)$ take $A=H_0$ and $B=H_t-H_0=:X$, 
remembering that 
\begin{equation}
X=\sum_i H_i
\end{equation}
 is a sum of local and trace-less terms.
Let us take as $H_0$ the operator 
\begin{equation}
	V=\otimes_{i=1}^n V_i 
\end{equation}
on an $n$-partite system, where the $V_i$ are trace-less and satisfy $V_i^2=\id$.
It follows that $\tr(VH_i)=0$ for every local operator $H_i$. The Gibbs state corresponding to $V$ is of the form 
\begin{equation}\omega_V = C(\id + \tanh(\beta)V)\end{equation} with 
some constant $C\in \RR$. Therefore, $\tr(\omega_VX)=0$ and using Lemma \ref{Peierls}, we get $F(\omega_V,V)\geq F(\omega_{H_t},H_t)$ for all $H_t\in\localhams(H_0)$, which proves the theorem. \\
This proof applies to all multipartite systems with even-dimensional local Hilbert-spaces. 
\end{proof}

We will use the existence of such \emph{passive states} identified in Thm.\ \ref{thm:passive} to show that TC is not universal in the scenario of restricted Hamiltonians $\localhams(H_0)$.}
\begin{theorem}[Non-universality of TC in the local setting]\label{thm:nonuniversalitylocal}
There exist $(\rho_0,H_0)$ such that for Hamiltonians restricted to {$\localhams(H_0)$}, TC is not universal for work extraction.
\end{theorem}
That is, protocols with thermalizing maps described by $\mc{M}_{\TO}$ (and hence the more general $\mc{M}_\GP$) can extract a strictly
larger amount of work than protocols using thermal contact. We have shown in Thm.\ \ref{thm:passive} that there exist initial configurations from which no work can be extracted in the local scenario using weak thermal contact. 
What is left to do is to discuss an example where work can be extracted from such an initial configuration using thermal operations. Consider a two-qubit system with initial Hamiltonian $V=\sigma_z\otimes \sigma_z$, where $\sigma_z$ is the Pauli-$Z$-matrix, and a maximally mixed initial state. This pair is a passive state for TC, as can be seen from the proof of Thm.~\ref{thm:passive}. One can now prove, using the ``thermo-majorisation'' condition \cite{Ruch,ThermoMaj,ThermoMaj2}, that for $\beta=1$ and $t\leq t_c:=\tanh^{-1}({\e^2-1}/({2\e^2}))\simeq 0.46$ it is possible to find a 
Gibbs-preserving map that maps $\mm$ to $\omega_{\sigma_z\otimes \sigma_z+t\id\otimes \sigma_z}$. The proof is, however, purely technical and provides no further insight. We therefore discuss it in Section \ref{sec:majorization}. Since the initial and final states are \emph{classical}, it follows that there also is a thermal operation yielding the same transition \cite{ThermoMaj2}. But since $\sigma_z\otimes \sigma_z+t \id\otimes \sigma_z\in\localhams(\sigma_z\otimes \sigma_z)$ is an allowed Hamiltonian, we can now use the optimal protocol. This allows us to extract an amount of work
\begin{equation}
\langle W\rangle^{{\localhams},\GP} (\mm,\sigma_z\otimes \sigma_z) = t\tanh(t)-\log(\cosh(t))>0,
\end{equation}
proving a gap between the extractable work using thermal operations and thermal contact.

\section{Technical results and proofs}

\subsection{General second law without control restrictions}\label{sec:aberg}
In this section we demonstrate that an arbitrary protocol that combines evolution under time-dependent Hamiltonians and Gibbs-preserving maps, that is, thermalizing maps in $\maps_\GP$, can extract at most $\Delta F(\rho,H)$ of work from the initial configuration. This has been shown before in Ref.\ \cite{Aberg14}, but we provide a proof here for convenience of the reader. Assume that the protocol consists of $N+1$ contacts to heat baths, let $(\rho_{j},H_{j})$ be the state and Hamiltonian right before we apply the $j$-th thermalizing map and $(\rho_0,H_0)$ the initial state and Hamiltonian, respectively. Right after applying a thermalizing map we have the quantum state $\sigma_{j} = \mc{G}_j(\rho_{j})$. We can always assume that the first step in the protocol is a Gibbs-preserving map, since the identity-map is also a Gibbs-preserving map. We have that $\rho_{j+1}=U_{j+1}\sigma_{j} U_{j+1}^{-1}$ where $U_{j+1}$ is the unitary that implements the time-dependent Hamiltonian dynamics starting after the $j$-th thermalizing map and ending right before the $j+1$-th thermal map. The work is given by
\begin{eqnarray}
\langle W\rangle^{\GP} &=& \sum_{j=0}^{N}
\left(
\tr\left(\sigma_{j} H_{j} - \rho_{j+1}H_{j+1}\right)
\right)
\\
&=& \sum_{j=0}^N \left(
F(\sigma_{j},H_{j})-F(\rho_{j+1},H_{j+1})\nonumber
\right),
\end{eqnarray}
where we have used that $\rho_{j+1}$ has the same entropy as $\sigma_{j}$. 
Let us re-order the sum to get 
\begin{widetext}
\begin{eqnarray}
\langle W \rangle^{\GP} &=& F(\sigma^0,H_0) - F(\rho_{N+1},H_{N+1}) + \sum_{j=1}^{N}
F(\sigma_{j},H_{j})-\sum_{j=0}^{N-1}F(\rho_{j+1},H_{j+1})
\nonumber\\
\label{eq:workbound:gp:free_energies} &=& F(\sigma_0,H_0) - F(\rho_{N+1},H_{N+1}) -\sum_{j=0}^{N-1}\left(F(\rho_{j+1},H_{j+1})-F(\sigma_{j+1},H_{j+1})\right).
\end{eqnarray}
\end{widetext}
We can now use that $\sigma_{j+1}=\mc{G}_{j+1}(\rho_{j+1})$ with $\mc{G}_{j+1}$ Gibbs-preserving relative to $H_{j+1}$: Let $\relent{\rho}{\sigma} = \tr(\rho\log \rho) - \tr(\rho \log \sigma)$ be the relative entropy. From the data-processing inequality, we get
\begin{equation}
0 \leq \relent{\rho_{j+1}}{\omega_{H_{j+1}}}-\relent{\sigma_{j+1}}{\omega_{H_{j+1}}} .
\end{equation}
{Since the relative entropy fulfils $\Delta F(\rho,H)~=~\relent{\rho}{\omega_H}/\beta$, we therefore obtain
\begin{eqnarray}
F(\rho_{j+1},H_{j+1})&-&F(\sigma_{j+1},H_{j+1})\nonumber\\ &=&\Delta F(\rho_{j+1},H_{j+1})-\Delta F(\sigma_{j+1},H_{j+1})\nonumber\\ &\geq& 0. 
\end{eqnarray}}
Inserting this into \eqref{eq:workbound:gp:free_energies}, finally yields
\begin{equation}
\langle W\rangle^{\GP}_{\OPT} \leq F(\rho,H)-F(\rho_{N+1},H) \leq \Delta F(\rho,H).
\end{equation}
With similar reasoning one obtains for a non-cyclic process from $p_0=(\rho_0,H_0)$ to $p_f=(\rho_f,H_f)$ the work-bound
\begin{equation}
\langle W\rangle^{\GP}_{\OPT}(p_0,p_f) \leq F(\rho_0,H_0)-F(\rho_f,H_f).
\end{equation}

\subsection{Classicality of the setting and passive states}
\label{sec:app:operator-norm}
We will {first} show that in the case of work-extraction using TC, the problem is largely classical, in the sense that the work bound of Thm.~1 can be achieved with protocols where $(\rho_t,H_t)$ are diagonal in the same basis at all times. In this {way}, we contribute to the discussion on the role of coherences in quantum thermodynamics.

To see this, {first note that the set $\hams(H_0)$ is compact and the infimum in the bound of Thm.~1 is indeed achieved. Furthermore, the set $\hams(H_0)$ is closed under conjugation by unitaries and the free energy fulfills 
\begin{equation}
	\Delta F(U\rho U^\dagger,H)=\Delta F(\rho,U^\dagger H U)
\end{equation}	
 for any unitary $U$. The penalty-term in Thm.~1 therefore simplifies to
\begin{align}\label{eq:app:normbound:infimum}
\inf_{\substack{H_t\in \hams(H_0)\\ \sigma\in \mathcal{U}[H_0](\rho_0)}}\!\!\!\!\!\!\Delta F(\sigma,H_t) =\!\!\! \min_{H_t\in \hams(H_0)}\!\!\!\!\Delta F(\rho_0,H_t).
\end{align}
In terms of the actual work-extraction protocol this means that we can choose as first step in the protocol a unitary dynamics that results in a mapping $(\rho_0,H_0)\mapsto(\rho_0,H^*)$. Here, $H^*$ is the Hamiltonian for which the minimum in Eq.~(\ref{eq:app:normbound:infimum}) is achieved. As explained in Section \ref{sec:aberg}, in the optimal protocol, this first step is followed by TC and after the TC an isothermal change of Hamiltonians from $H^*$ to $H^0$ is performed. This proccess has a work cost which is independent of the path of Hamiltonians, and given simply by $F(\omega(H^*),H^*)-F(\omega(H_0),H_0)$. Thus, within that process we can always choose Hamiltonians which are diagonal in the basis of $H^*$.
}
We will next show that the minimum in Eq.~(\ref{eq:app:normbound:infimum}) is always achieved for a Hamiltonian $H^*$ that is diagonal in the same basis as $\rho_0$.

\begin{lemma}[Passive states]\label{lemma:app:classical}
Let $\rho$ be a quantum state and $\mc{H}$ be a compact set of Hamiltonians that is closed under conjugation by arbitrary unitaries. Then there exists 
an $H^*\in\mc{H}$ which is diagonal in the same basis as $\rho$ and satisfies
\begin{align}
\Delta F(\rho,H^*) = \min_{H'\in\mc{H}}\Delta F(\rho,H').
\end{align}
Furthermore, $\rho$ is a passive state for $H^*$.
\end{lemma}

{\it Proof:} 
We show this by contradiction: Suppose is the minimum is only attained for some $H$ that is not diagonal in the same basis as $\rho$. We will show that we can lower the energy-expectation value by rotating $H$ into the eigenbasis of $\rho$, which reduces the free energy difference. Let $\lambda^\downarrow$ be the ordered spectrum of $H$ in non-increasing order, $p^\uparrow$ the spectrum of $\rho$ in non-decreasing order and let $d^\downarrow$ be the diagonal of $H$ in the eigenbasis of $\rho$ in non-increasing order (compared to $p^\uparrow$). Then $E=\tr(\rho H)=\sum_i p^\uparrow_i d^\downarrow_i$. It is easy to check that the vector $d^\downarrow$ is majorised by $\lambda^\downarrow$. Let us write partial sums as 
\begin{equation}
	S^i(\lambda^\downarrow) = \sum_{k=1}^i \lambda^\downarrow_k
\end{equation}
with $S^0(\lambda^\downarrow)=0$. Then $S^i(\lambda^\downarrow)\geq S^i(d^\downarrow)$ for all $i$. Now we have
\begin{align}
E' &:= \sum_i  p^\uparrow_i \lambda^\downarrow_i = \sum_i  p^\uparrow_i (S^i(\lambda^\downarrow)-S^{i-1}(\lambda^\downarrow)) \nonumber\\ 
&= p^\uparrow_nS^n(\lambda^\downarrow) - \sum_i^{n-1}S^i(\lambda^\downarrow)(p^\uparrow_{i+1}-p^\uparrow_i)\nonumber\\
&\leq  p^\uparrow_nS^n(\lambda^\downarrow) - \sum_i^{n-1}S^i(d^\downarrow)(p^\uparrow_{i+1}-p^\uparrow_i).
\end{align}
But $S^n(\lambda^\downarrow) = \tr(H) = S^n(d^\downarrow)$ and rearranging the sum back yields
\begin{align}
E' \leq \sum_i p^\uparrow_i d^\downarrow_i = E.
\end{align}
Hence, by rotating the Hamiltonian into the eigenbasis of $\rho$ we cannot increase the energy if we keep the ordering of the eigenvalues correct.
\qed


\subsection{Two-qubit example}\label{sec:majorization}
\label{sec:app:two-qubits}
\begin{figure}
\includegraphics[width=8cm]{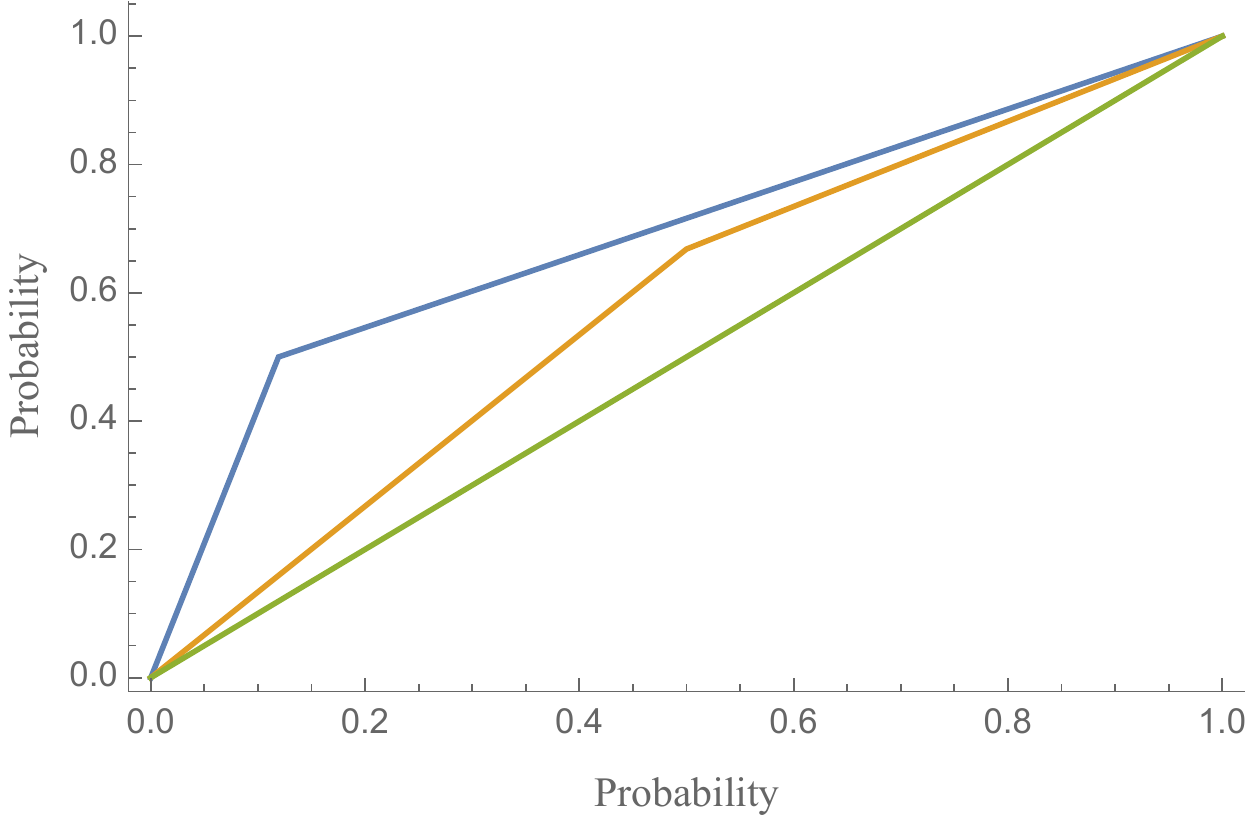}
\caption{The curves $g$ (in blue, see eq.~\eqref{eq:g}) and $f$ (in orange, see eq.~\eqref{eq:f}) illustrating the thermo-majorisation condition for the two qubit example. Since the orange curve (corresponding to the state $\omega_{\sigma_z\otimes \sigma_z+t\id \otimes \sigma_z}$ for $t< t_c$) lies below the blue curve (corresponding to the maximally mixed state) the transition from the maximally mixed state to $\omega_{\sigma_z\otimes \sigma_z+t\id \otimes \sigma_z}$ is possible using a thermal operation. The green curve (identity function) corresonds to the thermal state $\omega_{\sigma_z\otimes \sigma_z}$.}
\label{fig:maj}
\end{figure}

In this section we provide a detailed analysis of the example proving
\he{Thm.\ 5}. More precisely, we find an initial configuration $(\rho_0,H_0)$ of two qubits such that there exists a gap in work-extraction between protocols using thermalizing maps in $\maps_{\TC}$ and protocols using $\maps_\GP$.
At the heart of the argument will be majorisation theory in linear algebra \cite{Bhatia} {and its generalisation called d-majorisation or thermo-majorisation \cite{Ruch,ThermoMaj,ThermoMaj2}.}
Consider two qubits with initial configuration 
\begin{equation}
	(\rho_0,H_0)=(\mm, \sigma_z\otimes \sigma_z), 
\end{equation}
a maximally mixed state $\mm=\id/4$ on two qubits and a purely
interacting Hamiltonian. As shown in \he{Thm.\ 3}, no work can be extracted if $\hams=\hams_{\text{loc}}$ and the thermalizing maps are in $\maps_{\TC}$. We now present a protocol $\mc{P}^*$ which extracts a positive amount of work from $(\mm, \sigma_z\otimes \sigma_z)$, if $\hams=\hams_{\text{loc}}$ and thermalizing maps are in $\maps_\GP$. The first step in the protocol is to apply a GP-map $\mc{G}_0$ to the initial state $\mm$ to map it to some state $\sigma_0$. The state $\sigma_0$ is chosen such that $(\sigma_0,H_0)$ allows one to extract the difference of free energy $\Delta F(\sigma_0,H_0)$ as work just by using $\mc{P}_{\TC}^*$,
the optimal protocol employing only thermalizing maps in $\maps_{\TC}$. 
This protocol can be implemented on $(\sigma_0,H_0)$ because $\sigma_0$ is chosen so that it is the Gibbs state of an accessible Hamiltonian, that is, 
\begin{eqnarray}\label{eq:sigma0}
&&\sigma_0=\omega_{\tilde{H}}\quad \text{with}\quad \tilde{H} \in \localhams(\sigma_z\otimes \sigma_z).
\end{eqnarray}
In the following, we first introduce a state $\sigma_0$ fulfilling \eqref{eq:sigma0}, and then show that there exists a GP-map such that 
\begin{equation}
\mc{G}_0(\mm)=\sigma_0.
 \end{equation}
Consider the Gibbs state $\omega_{\sigma_z\otimes \sigma_z + t \id\otimes \sigma_z}$.  It is diagonal in the same basis as $\omega_{\sigma_z\otimes \sigma_z}$ and has the ordered spectrum (from low energy to high energy)
\begin{equation}
\omega_2=\frac{1}{2(1+\e^2)}(\e^2f^+(t)),\e^2f^{-}(t)),f^+(t),f^{-}(t)), 
\end{equation}
with $f^\pm(t) = 1\pm \tanh(t)$, while $\omega_{\sigma_z\otimes \sigma_z}$ has the ordered spectrum
\begin{equation}
\omega_1=\frac{1}{2(1+\e^2)}(\e^2,\e^2,1,1). 
\end{equation}
Since all states and Hamiltonians are diagonal in the same basis, we can use the thermo-majorisation (or $d$-majorisation) condition \cite{ThermoMaj,ThermoMaj2,Ruch} 
to decide whether we can map $\mm$, with spectrum $p=(1,1,1,1)/4$, 
to $\omega_{\sigma_z\otimes \sigma_z + t \id\otimes \sigma_z}$ by a Gibbs-preserving operation.
To evaluate thermo-majorisation we have to order the vectors with entries $r_i=(\omega_2)_i/(\omega_1)_i$ and $r_i'=p_i/(\omega_1)_i$ in non-increasing order. Let $\sigma,\sigma'$ be the permutations that do this, i.e.,
 \begin{equation}
r_{\sigma(1)}\geq \dots \geq  r_{\sigma(4)}
 \end{equation}
 and similarly for $p$ and $\sigma'$. The vectors $r,r'$ are given by
 \begin{eqnarray}
	r &=& (1+ \tanh(t),1- \tanh(t),1+ \tanh(t),1- \tanh(t)),\nonumber\\
	\\ r'&=&\frac{1+\e^2}{2}(\e^{-2},\e^{-2},1,1).
 \end{eqnarray}
Thus, our permutations are given by (for example)
\begin{align}
\sigma = \left(\begin{array}{cccc} 1 & 2 & 3 & 4 \\ 3& 1 & 4 & 2 \end{array}\right),\quad \sigma' = \left(\begin{array}{cccc} 1 & 2 & 3 & 4 \\ 3& 4 & 1 & 2 \end{array}\right).
\end{align}
In particular note that $\sigma\neq \sigma'$.  
Thermo-majorisation implies that $\id/4$ can be mapped to $\omega_{\sigma_z\otimes \sigma_z+t \id\otimes \sigma_z}$ iff the curve $g$ of straight lines connecting the points with coordinates
\begin{align}
\label{eq:g}
\left(\sum_{j=1}^k (\omega_1)_{\sigma'(j)},\sum_{j=1}^k p_{\sigma'(j)} \right),\quad k=1,\ldots,4, 
\end{align}
lies above the curve of straight lines $f$ connecting the points with coordinates
\begin{align}
\label{eq:f}
\left(\sum_{j=1}^k  (\omega_1)_{\sigma(j)},\sum_{j=1}^k  (\omega_2)_{\sigma(j)} \right),\quad k=1,\ldots,4. 
\end{align}
It follows quickly that the condition holds for the first points. For the second point, a calculation shows that 
\begin{equation}
	g(1/2) = \frac{1}{2}\left(1+\frac{\e^2-1}{2\e^2}\right)  
\end{equation}
while 
\begin{equation}
	f(1/2) = \frac{1}{2}(1+\tanh(t)). 
\end{equation}
Thus, from the second points we get the condition 
\begin{equation}
	t\leq t_c:=\tanh^{-1}\left(\frac{\e^2-1}{2\e^2}\right)\simeq 0.46. 
\end{equation}
It turns out that there are no further constraints (see Fig.\ \ref{fig:maj}).
Hence, if we choose 
\begin{equation}
	\sigma_0=\omega_{\sigma_z\otimes \sigma_z + t \id\otimes \sigma_z}
\end{equation}	
	 for any $t\leq t_c$ it fulfills \eqref{eq:sigma0} by definition and also there exists a corresponding map $\mc{G}_{0} \in \maps_{\TO}(H_0)\subset \maps_\GP(H_0)$ sending $\mm$ to $\sigma_0$.
What is left is to compute the free energy-difference, which equals the total extracted work by subsequently applying $\mc{P}_{\TC}^*$ to $(\sigma_0,H_0)$.
This is easily done and yields
\begin{eqnarray}
\nonumber \langle W \rangle^{\TC} \left(\mc{P}^*,(\mm,\sigma_z\otimes \sigma_z)\right)&=&\langle W \rangle \left(\mc{P}^*_{\TC},(\sigma_0,\sigma_z\otimes \sigma_z)\right) \nonumber\\
\nonumber &=& \Delta F(\sigma_0,\sigma_z\otimes \sigma_z)\nonumber \\ 
 &=& t\tanh(t)-\log(\cosh(t)),\nonumber\\
\end{eqnarray}
which is positive for $t\neq 0$ and bounded by from above by 
\begin{eqnarray}
\Delta F(\Omega,\sigma_z\otimes \sigma_z) = \log(\cosh(1))
\end{eqnarray}
for $|t|\leq t_c$.

\section{Conclusion}
In this work, we have established a versatile framework of thermodynamic operations under lack of experimental control. We have seen that one of the key surprising results of quantum thermodynamics, namely that weak thermal contact already allows to extract all the work that could possibly be extracted, ceases to be valid under such simple and natural constraints. This shows that operational restrictions cannot be considered independently, because they can interact in a non-trivial way.
Our results point into the direction that quantum thermodynamics is significantly more complex whenever such 
ubiquitous
constraints are present and hope that this work initiates further research in this direction.\\

\section{Acknowledgements}

We acknowledge funding from the A.-v.-H., the BMBF, the 
 EU (RAQUEL, SIQS, COST, AQuS), the ERC (TAQ)\he{, the COST network,} and the Studienstiftung des Deutschen Volkes.


\begin{thebibliography}{99}

	
\bibitem{Alicki}
	R.\ Alicki, J.\ Phys.\ A {\bf 12}, L103 (1979).
	
\bibitem{Alicki2}
R. Alicki, M. Horodecki, P. Horodecki and R. Horodecki. {Open Syst. Inf. Dyn.} {\bf 11,} 205 (2004).

\bibitem{Esposito}
M. Esposito and C. Van den Broeck. {Europhys. Lett.} {\bf 95,} 40004 (2011).

\bibitem{Linden10}
	N.\ Linden, S.\ Popescu and P.\ Skrzypczyk, 
	Phys.\ Rev.\ Lett {\bf 105,} 130401 (2010).
\bibitem{Brunner12}
	N.\ Brunner, N.\ Linden, S.\ Popescu and P.\ Skrzypczyk, 
	{Phys.\ Rev.\ E} {\bf 85,} 05111 (2012).
	

\bibitem{Anders13}
	J.\ Anders and V.\ Giovannetti, 
	{New J.\ Phys.} {\bf 15}, 033022 (2013).

\bibitem{OOE}
J.\ Eisert, M.\ Friesdorf, and C.\ Gogolin,
Nature Phys.\ {\bf 11}, 124 (2015).

\bibitem{Jarzynski97}
C. Jarzynski, {Phys. Rev. Lett.} {\bf 78,} 2690 (1997).


\bibitem{Egloff12}
	D.\ Egloff, O.\ C.\ O.\ Dahlsten, R.\ Renner and V.\ Vedral, 
	New J.\ Phys.\ {\bf 17},  073001 (2015). 

\bibitem{Aberg13}
	J.\ Aberg,\ Nature Comm.\ {\bf 4}, 1925 (2013).


\bibitem{Gallego2013}
  		R.\ Gallego, A.\  Riera,  and J.\ Eisert, Thermal machines beyond the weak coupling regime,
  		\je{New J.\ Phys.\ {\bf 16}, 125009 (2014).}

\bibitem{ThermoMaj2}
	M.\ Horodecki and J.\ Oppenheim, Nature Comm.\ {\bf 4}, 2059 (2013).
	
\bibitem{Resource}
	F.\ G.\ S.\ L.\ Brandao, M.\ Horodecki, J.\ Oppenheim, 
	J.\ M.\ Renes, and R.\ W.\ Spekkens, Phys. Rev. Lett. {\bf 111}, 250404 (2013).
	
\bibitem{SecondLaws} F.\ G.\ S.\ L.\ Brandao, M.\ Horodecki, N.\ H.\ Y.\ Ng, J.\ Oppenheim,	
	and S.\ Wehner, PNAS {\bf 112}, 3275 (2015).
	
\bibitem{SL2}
	N.\ H.\ Y.\ Ng, L.\ Mancinska, C.\ Cirstoiu, J.\ Eisert, and S.\ Wehner, 
		New J.\ Phys.\ {\bf 17}, 085004 (2015).

\bibitem{ThermoMaj}
 	H.\ D.\ Janzing, P.\ Wocjan, R.\ Zeier, R.\ Geiss, and T.\ Beth,
	Int.\ J.\ Th.\ Phys.\ {\bf 39}, 2717
	(2000).	
	

\bibitem{Renes14}
	J.\ M.\ Renes, {Euro. Phys. J. Plus} {\bf 129}, 7 (2014) 


\bibitem{RossnagelComment}	
	{This observation is related in spirit with
Ref. \cite{Rossnagel}. However, there control is used to couple the working system to a engineered non-equilibrium bath. Our result is different in the sense that it proves control advantageous for thermal baths in equilibrium.}



\bibitem{Rossnagel}	
	J. Rossnagel, O. Abah, F. Schmidt-Kaler, K. Singer, and E. Lutz, {Phys. Rev. Lett.} {\bf 112,} 030602 (2014).


\bibitem{PathConnected} 
        {We will {assume} that $\hams(H_0)$ is path-connected, which as we will see later is naturally {fulfilled} in relevant specific limitations on $\hams$.}
        
\bibitem{footnote1}
Using time-dependent Hamiltonian dynamics to model work-extraction is standard in the literature on quantum thermodynamics {\cite{Jarzynski97,Jarzynski99,Gallego2013,Anders13}}.
{It} is, however, an ongoing debate whether measuring the work done on in such a process by 
the energy expectation is {fully} justified \cite{Gallego2015}. We have chosen this work-measure as it simplifies the analysis in this work considerably and we expect that similar effects as our main results will remain true also for other work-measures.


\bibitem{Gallego2015}
  R.\ Gallego, J.\ Eisert, and H.\ Wilming, Thermodynamic work from operational principles, arXiv:1504.05056.

\bibitem{Jarzynski99}
	C.\ Jarzynski,\ J.\ Stat.\ Phys. {\bf 96}, 415 (1999).

	
\bibitem{Riera12}
	A.\ Riera, C.\ Gogolin, and J.\ Eisert,
	Phys.\ Rev.\ Lett.\ {\bf 108,} 080402 (2012).
	


\bibitem{Faist2014} 
	P.\ Faist, J.\ Oppenheim, and R.\ Renner,  New J.\ Phys. {\bf 17}, 043003 (2015).


\bibitem{Aberg14}
        J.\ Aberg, {Phys.\ Rev.\ Lett.} {\bf 113}, 150402 (2014).


\bibitem{MaxEnt}
\he{This principle states that {the} extracted work is maximal for a reversible process initially in an equilibrium state. Here we are concerned with non-equilibrium initial state. Nevertheless it is indeed the case that the optimal protocol under TC maps is reversible. However, {the} extended family of maps given by the Gibbs-preserving condition \eqref{eq:gibbspreserving} may {also be} reversible, so {that} there is indeed no reason why one could not conceive a reversible protocol allowing more work if allowing for more general operations than TC. Altogether, Observation 2 is not equivalent {with} the maximum-work principle.}

\bibitem{Allahverdyan}
A. E. Allahverdyan and Th. M. Nieuwenhuizen, {Phys. Rev.} E {\bf 71,} 046107 (2005).


\bibitem{PeierlsBogoliubov}
	D.\ Ruelle, {\it Statistical mechanics: Rigorous results} (World Scientific, 1969).	
		
\bibitem{Ruch}
        E.\ Ruch, R.\ Schranner, and T.\ H.\ Seligman, 
        J.\ Chem.\ Phys.\ {\bf 69},  386 (1978).
	
\bibitem{Bhatia}
    R.\ Bhatia, {\it Matrix analysis} (Springer, Berlin, 1997).


	
	
  	

\end{thebibliography}
\end{document}